\documentclass{article}

\usepackage[pdftex,paper=a4paper,left=3cm,right=3cm,top=4cm,bottom=4cm]{geometry}
\usepackage{amsmath, amssymb, amsthm}
\usepackage{xspace}
\usepackage{mathrsfs}
\usepackage{graphicx}
\usepackage{algorithmic}
\usepackage{algorithm}
\usepackage{color}
\usepackage{caption}
\usepackage{subcaption}

\usepackage[pdftex,linktocpage,plainpages=false,breaklinks,hypertexnames=false]{hyperref}
 
\newcommand{\NN}{\mathbb{N}}
\newcommand{\ZZ}{\mathbb{Z}}
\newcommand{\RR}{\mathbb{R}}

\renewcommand{\d}{\mathrm{d}}
\newcommand{\T}{\mathrm{T}}
\newcommand{\N}{\mathcal{N}}
\newcommand{\NULL}{0}
\newcommand{\SET}[1]{\left\{#1\right\}}
\newcommand{\Ex}[2][]{\mbox{\rm\bf E}_{#1}\hspace{-0.03cm}\left[#2\right]}
\renewcommand{\Pr}[2][]{\mbox{\rm\bf Pr}_{#1}\hspace{-0.03cm}\left[#2\right]}
\newcommand{\A}{\mathbb{A}}
\newcommand{\B}{\mathbb{B}}
\newcommand{\E}{\mathbb{E}}
\newcommand{\F}{\mathbb{F}}

\newcommand{\DOT}{\,.}
\newcommand{\COMMA}{\,,}
\newcommand{\WHERE}{\,\colon\,}
\newcommand{\DEF}{\mathop{:=}}
\newcommand{\FED}{\mathop{=:}}
\newcommand{\diag}{\mathrm{diag}}
\newcommand{\SPAN}[1]{\mathrm{span}\{#1\}}

\newcommand{\e}{\varepsilon}
\newcommand{\ID}[1]{\mathbb{I}_{#1}}
\newcommand{\ZERO}[2]{\mathbb{O}_{#1\times#2}}

\providecommand{\qedhere}{\tag*{\qed}}

\newtheorem{theorem}{Theorem}
\newtheorem{lemma}[theorem]{Lemma}

\newtheorem{definition}[theorem]{Definition}
\newtheorem{corollary}[theorem]{Corollary}

\providecommand{\institute}[1]{\date{#1}}
\providecommand{\email}[1]{\texttt{#1}}

\begin{document}

\title{Finding Short Paths on Polytopes by the\\Shadow Vertex Algorithm\thanks{This research was supported by ERC Starting Grant 306465 (BeyondWorstCase).}}

\author{Tobias Brunsch \and Heiko R\"oglin}

\institute{
Department of Computer Science\\
University of Bonn,
Germany\\
\email{\small\{brunsch,roeglin\}@cs.uni-bonn.de}}

\maketitle

\begin{abstract}
We show that the shadow vertex algorithm can be used to 
compute a short path between a given pair of vertices of a polytope $P = \SET{ x \in \RR^n \WHERE Ax \leq b }$ along the edges of~$P$, where $A \in \RR^{m \times n}$. Both, the length of the path and the running time of the algorithm, are polynomial in~$m$, $n$, and a parameter~$1/\delta$ that is a measure for the flatness of the vertices of~$P$. For integer matrices $A \in \ZZ^{m \times n}$ we show a connection between~$\delta$ and the largest absolute value~$\Delta$ of any sub-determinant of~$A$, yielding a bound of $O(\Delta^4 m n^4)$ for the length of the computed path. This bound is expressed in the same parameter~$\Delta$ as the recent non-constructive bound of $O(\Delta^2 n^4 \log (n \Delta))$ by Bonifas et al.~\cite{BonifasDEHN12}.

For the special case of totally unimodular matrices, the length of the computed path simplifies to $O(m n^4)$, which significantly improves the previously best known constructive bound of $O(m^{16} n^3 \log^3 (mn))$ by Dyer and Frieze~\cite{DyerF94}.
\end{abstract}

\section{Introduction}

We consider the following problem: Given a matrix $A = [a_1, \ldots, a_m]^\T \in \RR^{m \times n}$, a vector $b \in \RR^m$, and two vertices~$x_1$ and~$x_2$ of the polytope $P \!=\! \SET{ x \in \RR^n \WHERE \! Ax \leq b }$, find a short path from~$x_1$ to~$x_2$ along the edges of~$P$ efficiently. In this context \emph{efficient} means that the running time of the algorithm is polynomially bounded in~$m$, $n$, and the length of the path it computes. Note, that the polytope~$P$ does not have to be bounded.

The \emph{diameter~$d(P)$} of the polytope~$P$ is the smallest integer~$d$ that bounds the length of the shortest path between any two vertices of~$P$ from above. The \emph{polynomial Hirsch conjecture} states that the diameter of~$P$ is polynomially bounded in~$m$ and~$n$ for any matrix~$A$ and any vector~$b$. As long as this conjecture remains unresolved, it is unclear whether there always exists a path of polynomial length between the given vertices~$x_1$ and~$x_2$. Moreover, even if such a path exists, it is open whether there is an efficient algorithm to find it.

\paragraph{Related work} The diameter of polytopes has been studied extensively in the last decades. In 1957 Hirsch conjectured that the
diameter of~$P$ is bounded by $m-n$  for any matrix~$A$ and any vector~$b$ (see Dantzig's seminal book about linear programming~\cite{Dantzig63}). This conjecture has been disproven by Klee and Walkup~\cite{KleeW67} who gave an unbounded counterexample. However, it remained open for quite a long time whether the conjecture holds for bounded polytopes. More than fourty years later
Santos~\cite{Santos10} gave the first counterexample to this refined conjecture showing that there are bounded polytopes~$P$ for which $d(P) \geq (1+\e)
\cdot m$ for some $\e > 0$. This is the best known lower bound today. On the other hand, the best known upper bound of $O(m^{1+\log n})$ due
to Kalai and Kleitman~\cite{KalaiK92} is only quasi-polynomial. It is still an open question whether $d(P)$ is always polynomially bounded 
in~$m$ and~$n$. This has only been shown for special classes of polytopes like $0/1$ polytopes, flow-polytopes, and  the transportation
polytope. For these classes of polytopes bounds of~$m-n$ (Naddef~\cite{Naddef89}), $O(mn \log n)$ (Orlin~\cite{Orlin97}), and $O(m)$ (Brightwell et al.~\cite{BrightwellHS06}) have been shown, respectively. On the other hand, there are bounds on the diameter of far more general classes of polytopes that
depend polynomially on~$m$, $n$, and on additional parameters. Recently, Bonifas et al.~\cite{BonifasDEHN12} showed that the diameter
of polytopes~$P$ defined by integer matrices~$A$ is bounded by a polynomial in~$n$ and a parameter that depends on the matrix~$A$. They showed that $d(P) = O(\Delta^2 n^4 \log (n\Delta))$, where~$\Delta$ is the largest absolute value of any
sub-determinant of~$A$. Although the parameter~$\Delta$ can be very large in general, this approach allows to obtain bounds for classes of polytopes for which~$\Delta$ is known to be small. For example, if the matrix~$A$ is totally unimodular, i.e., if all sub-determinants
of~$A$ are from $\SET{ -1, 0, 1 }$, then their bound simplifies to $O(n^4 \log n)$, improving the previously best known bound of $O(m^{16}
n^3 \log^3 (mn))$ by Dyer and Frieze~\cite{DyerF94}.

We are not only interested in the existence of a short path between two vertices of a polytope but we want to compute such a path efficiently. It is clear that lower bounds for the diameter of polytopes have direct (negative) consequences for this algorithmic problem. However, upper bounds
for the diameter do not necessarily have algorithmic consequences as they might be non-constructive. 
The aforementioned bounds of Orlin, Brightwell et al., and Dyer and Frieze are constructive, whereas the bound of Bonifas et al.\ is not.

\paragraph{Our contribution} We give a constructive upper bound for the diameter of the polytope $P = \SET{ x \in \RR^n \WHERE Ax \leq b }$ for arbitrary matrices $A \in \RR^{m \times n}$ and arbitrary vectors $b \in \RR^m$.\footnote{Note that we do not require the polytope to be bounded.} This bound is polynomial in~$m$, $n$, and a parameter~$1/\delta$, which depends only on the matrix~$A$ and is a measure for the angle between edges of the polytope~$P$ and their neighboring facets. 
We say that a facet~$F$ of the polytope~$P$ is neighboring an edge~$e$ if exactly one of the endpoints of~$e$ belongs to~$F$.
The parameter~$\delta$ denotes the smallest sine of any angle between an edge and a neighboring facet in~$P$.
If, for example, every edge is orthogonal to its neighboring facets, then~$\delta=1$. On the other hand, if there exists an edge that
is almost parallel to a neighboring facet, then~$\delta\approx 0$. 
The formal definition of~$\delta$ is deferred to Section~\ref{sec:Parameter}.

A well-known pivot rule for the simplex algorithm is the shadow vertex rule, which gained attention in recent years because
it has been shown to have polynomial running time in the framework of smoothed analysis~\cite{SpielmanTeng}. We will present
a randomized variant of this pivot rule that computes a path between two given vertices of the polytope~$P$.
We will introduce this variant in Section~\ref{sec:Algo} and we call it \emph{shadow vertex algorithm} in the following.

\begin{theorem}
\label{maintheorem}
Given vertices~$x_1$ and~$x_2$ of~$P$, the shadow vertex algorithm efficiently computes a path from~$x_1$ to~$x_2$ on
the polytope~$P$ with expected length~$O \big(\! \frac{mn^2}{\delta^2} \!\big)$.
\end{theorem}

% Let us emphasize that the algorithm is very simple and efficient in the sense that its running time depends polynomially on~$m$, $n$
% and the length of the path it computes.
Let us emphasize that the algorithm is very simple and its running time depends only polynomially on~$m$, $n$
and the length of the path it computes.

Theorem~\ref{maintheorem} does not resolve the polynomial Hirsch conjecture as the value~$\delta$ can be exponentially small. Furthermore, it does not imply a good running time of the shadow vertex method for optimizing linear programs because for the variant considered in this paper both vertices have to be known. Contrary to this, in the optimization problem the objective is to determine the optimal vertex.
%We think that our bound significantly overestimates the running time of the simplex algorithm since in any of its step we assume the worst-case to happen, i.e., that the polytope around the current vertex is pretty flat. Nevertheless, expressing the diameter of polytope~$P$ in parameter~$\delta$ seems to be more natural than expressing it in the parameter~$\Delta$ of Bonifas et al.\ because our bound holds for any matrix $A \in \RR^{m \times n}$. The definition of parameter~$\Delta$ can also be extended to matrices $A \in \RR^{m \times n}$ but it becomes meaningless with regard to bounding the diameter of~$P$ when considering non-integral matrices since scaling of rows of~$A$ results in a different value of~$\Delta$.
To compare our results with the result by Bonifas et al.~\cite{BonifasDEHN12},
 we show that, if~$A$ is an integer matrix, then $\frac{1}{\delta} \leq n \cdot \Delta^2$, which yields the following corollary.

\begin{corollary}
\label{maincorollaryI}
Let $A \in \ZZ^{m \times n}$ be an integer matrix and let $b \in \RR^m$ be a real-valued vector. Given vertices~$x_1$ and~$x_2$ of~$P$,
the shadow vertex algorithm efficiently computes a path from~$x_1$ to~$x_2$ on
the polytope~$P$ with expected length $O(\Delta^4 mn^4)$.
\end{corollary}

This bound is worse than the bound of Bonifas et al., but it is constructive. Furthermore, if~$A$ is a totally unimodular matrix, then $\Delta = 1$. Hence, we obtain the following corollary.

\begin{corollary}
\label{maincorollaryII}
Let $A \in \ZZ^{m \times n}$ be a totally unimodular matrix and let $b \in \RR^m$ be a vector. Given vertices~$x_1$ and~$x_2$ of~$P$,  the shadow vertex algorithm efficiently computes a path from~$x_1$ to~$x_2$ on
the polytope~$P$ with expected length $O(mn^4)$.
\end{corollary}

This is a significant improvement upon the previously best known constructive bound of $O(m^{16}n^3 \log^3 (mn))$ due to Dyer and Frieze because we can assume~$m\ge n$. Otherwise,~$P$ does not have vertices and the problem is ill-posed.

\paragraph{Organization of the paper} In Section~\ref{sec:Algo} we describe the shadow vertex algorithm.
In Section~\ref{sec:Outline} we give an outline of our analysis and present the main ideas.
After that, in Section~\ref{sec:Parameter}, we introduce the parameter~$\delta$ and discuss some of its properties. Section~\ref{sec:analysis} is devoted to the proof of Theorem~\ref{maintheorem}.
%Due to space limitations some proofs are deferred to Section~\ref{sec:deferred proofs}.
The probabilistic foundations of our analysis are provided in Section~\ref{sec:protheory}.

\section{The Shadow Vertex Algorithm}
\label{sec:Algo}

Let us first introduce some notation. For an integer $n \in \NN$ we denote by~$[n]$ the set $\SET{ 1, \ldots, n }$. Let $A \in \RR^{m \times n}$ be an $m \times n$-matrix and let $i \in [m]$ and $j \in [n]$ be indices. With~$A_{i,j}$ we refer to the $(m-1) \times (n-1)$-submatrix obtained from~$A$ by removing the $i^\text{th}$ row and the $j^\text{th}$ column. We call the determinant of any $k \times k$-submatrix of~$A$ a \emph{sub-determinant of~$A$ of size~$k$}. By~$\ID{n}$ we denote the $n \times n$-identity matrix $\diag(1, \ldots, 1)$ and by $\ZERO{m}{n}$ the $m \times n$-zero matrix. If $n \in \NN$ is clear from the context, then we define vector~$e_i$ to be the $i^\text{th}$ column of~$\ID{n}$. For a vector $x \in \RR^n$ we denote by $\|x\| = \|x\|_2$ the Euclidean norm of~$x$ and by $\N(x) = \frac{1}{\|x\|} \cdot x$ for $x \neq \NULL$ the normalization of vector~$x$.

\subsection{Shadow Vertex Pivot Rule}

Our algorithm is inspired by the shadow vertex pivot rule for the simplex algorithm. 
Before describing our algorithm, we will briefly explain the geometric intuition behind this pivot rule.
For a complete and more formal description, we refer the reader to~\cite{Borgwardt86} or~\cite{SpielmanTeng}. 
Let us consider the linear program~$\min c^\T x$ subject to~$x\in P$ for some vector~$c\in\RR^n$ and assume 
that an initial vertex~$x_1$ of the polytope~$P$ is known. For the sake of simplicity, we assume that
there is a unique optimal vertex~$x^{\star}$ of~$P$ that minimizes the objective function~$c^\T x$.
The shadow vertex pivot rule first computes a vector~$w\in\RR^n$ such that the vertex~$x_1$ 
minimizes the objective function~$w^\T x$ subject to~$x\in P$. Again for the sake of simplicity, let
us assume that the vectors~$c$ and~$w$ are linearly independent.

In the second step, the polytope~$P$ is projected onto the plane spanned by the vectors~$c$ and~$w$.
The resulting projection is a polygon~$P'$ and one can show that the projections 
of both the initial vertex~$x_1$ and the optimal vertex~$x^{\star}$ are vertices of this polygon.
Additionally every edge between two vertices~$x$ and~$y$ of~$P'$ corresponds to an edge of~$P$ between
two vertices that are projected onto~$x$ and~$y$, respectively. Due to these properties a path from the projection of~$x_1$ to
the projection of~$x^{\star}$ along the edges of~$P'$ corresponds to a path from~$x_1$ to~$x^{\star}$ 
along the edges of~$P$.

This way, the problem of finding a path from~$x_1$ to~$x^{\star}$ on the polytope~$P$ is reduced to
finding a path between two vertices of a polygon. There are at most two such paths and the shadow vertex
pivot rule chooses the one along which the objective~$c^\T x$ improves.
%(except for some special cases that we ignore in this intuitive description this choice is unique)

\subsection{Our Algorithm}
\label{sec:our algorithm}

As described in the introduction we consider the following problem: We are given a matrix $A = [a_1, \ldots, a_m]^\T \in \RR^{m \times n}$, a vector $b \in \RR^m$, and two vertices $x_1, x_2$ of the polytope $P = \SET{ x \in \RR^n \WHERE Ax \leq b }$. Our objective is to find a short path from~$x_1$ to~$x_2$ along the edges of~$P$.

We propose the following variant of the shadow vertex pivot rule to solve this problem:
First choose two vectors~$w_1,w_2\in\RR^n$ such that~$x_1$ uniquely minimizes~$w_1^\T x$ subject to~$x\in P$
and~$x_2$ uniquely maximizes~$w_2^\T x$ subject to~$x\in P$. Then project the polytope onto the plane
spanned by~$w_1$ and~$w_2$ in order to obtain a polygon~$P'$. Let us call the projection~$\pi$. By the same
arguments as for the shadow vertex pivot rule, it follows that~$\pi(x_1)$ and~$\pi(x_2)$ are vertices of~$P'$
and that a path from~$\pi(x_1)$ to~$\pi(x_2)$ along the edges of~$P'$ can be translated into a path
from~$x_1$ to~$x_2$ along the edges of~$P$. Hence, it suffices to compute such a path to solve the problem.
Again computing such a path is easy because~$P'$ is a two-dimensional polygon. 

The vectors~$w_1$ and~$w_2$ are not uniquely determined, but they can be chosen from cones that are determined
by the vertices~$x_1$ and~$x_2$ and the polytope~$P$. We choose~$w_1$ and~$w_2$ randomly from these cones.
A more precise description of this algorithm is given as Algorithm~\ref{algorithm:SV}.
 
\begin{algorithm*}[th]
  \caption{Shadow Vertex Algorithm}
  \label{algorithm:SV}
  \begin{algorithmic}[1]

    \STATE Determine~$n$ linearly independent rows $u_k^\T$ of~$A$ for which $u_k^\T x_1 = b_k$. \label{line:opt vectors I}

    \STATE Determine~$n$ linearly independent rows $v_k^\T$ of~$A$ for which $v_k^\T x_2 = b_k$. \label{line:opt vectors II}

    \STATE Draw vectors $\lambda, \mu \in (0, 1]^n$ independently and uniformly at random. \label{line:randomness}

    \STATE Set $w_1 = -\left[ \N(u_1), \ldots, \N(u_n) \right] \cdot \lambda$ and $w_2 = \left[ \N(v_1), \ldots, \N(v_n) \right] \cdot \mu$. \label{line:opt directions}

    \STATE Use the function $\pi: x \mapsto \big( w_1^\T x, w_2^\T x \big)$ to project~$P$ onto the Euclidean plane and obtain the shadow vertex polygon $P' = \pi(P)$. \label{line:shadow vertex polytope}

    \STATE Walk from~$\pi(x_1)$ along the edges of~$P'$ in increasing direction of the second coordinate until~$\pi(x_2)$ is found. \label{line:walk}

    \STATE Output the corresponding path of~$P$.

  \end{algorithmic}
\end{algorithm*}

Let us give some remarks about the algorithm above.
The vectors $u_1, \ldots, u_n$ in Line~\ref{line:opt vectors I} and the vectors $v_1, \ldots, v_n$ 
in Line~\ref{line:opt vectors II} must exist because~$x_1$ and~$x_2$ are vertices of~$P$.
The only point where our algorithm makes use of randomness is in Line~\ref{line:randomness}. 
By the choice of~$w_1$ and~$w_2$ in Line~\ref{line:opt directions}, $x_1$ is the unique optimum of the 
linear program $\min w_1^\T x$ s.t.\ $x \in P$ and~$x_2$ is the unique optimum of the linear
program $\max w_2^\T x$ s.t.\ $x \in P$. The former follows because for any~$y\in P$ with~$y\neq x_1$
there must be an index~$k\in[n]$ with~$u_k^\T x_1 < b_k$. The latter follows analogously. 
Note, that
$
  \|w_1\|
  \leq \sum_{k=1}^n \lambda_k \cdot \|\N(u_k)\|
  \leq \sum_{k=1}^n \lambda_k
  \leq n
$
and, similarly, $\|w_2\| \leq n$. The shadow vertex polygon~$P'$ in Line~\ref{line:shadow vertex polytope} has several important properties: The projections of~$x_1$ and~$x_2$ are vertices of~$P'$ and all edges of~$P'$ correspond to projected edges of~$P$. Hence, any path on the edges of~$P'$ is the projection of a path on the edges of~$P$. Though we call~$P'$ a polygon, it does not have to be bounded. This is the case if~$P$ is unbounded in the directions~$w_1$ or~$-w_2$. Nevertheless, there is always a path from~$x_1$ to~$x_2$ which will be found in Line~\ref{line:walk}.
For more details about the shadow vertex pivot rule and formal proofs of these properties, we refer to the book of Borgwardt~\cite{Borgwardt86}.

To give a bit intuition why these statements hold true, consider the projection depicted in Figure~\ref{fig:shadow vertex}. We denote the first coordinate of the Euclidean plane by~$\xi$ and the second coordinate by~$\eta$. Since~$w_1$ and~$w_2$ are chosen such that~$x_1$ and~$x_2$ are, among the points of~$P$, optimal for the function $x \mapsto w_1^\T x$ and $x \mapsto w_2^\T x$, respectively, the projections~$\pi(x_1)$ and~$\pi(x_2)$ of~$x_1$ and~$x_2$ must be the leftmost vertex and the topmost vertex of $P' = \pi(P)$, respectively. As~$P'$ is a (not necessarily bounded) polygon, this implies that if we start in vertex~$\pi(x_1)$ and follow the edges of~$P'$ in direction of increasing values of~$\eta$, then we will end up in~$\pi(x_2)$ after a finite number of steps. This is not only true if~$P'$ is bounded (as depicted by the dotted line and the dark gray area) but also if~$P$ is unbounded (as depicted by the dashed lines and the dark gray plus the light gray area). 
Moreover, note that the slopes of the edges of the path from~$\pi(x_1)$ to~$\pi(x_2)$ are positive and monotonically decreasing.

\begin{figure}
  \begin{center}
    \includegraphics[width=0.3\textwidth]{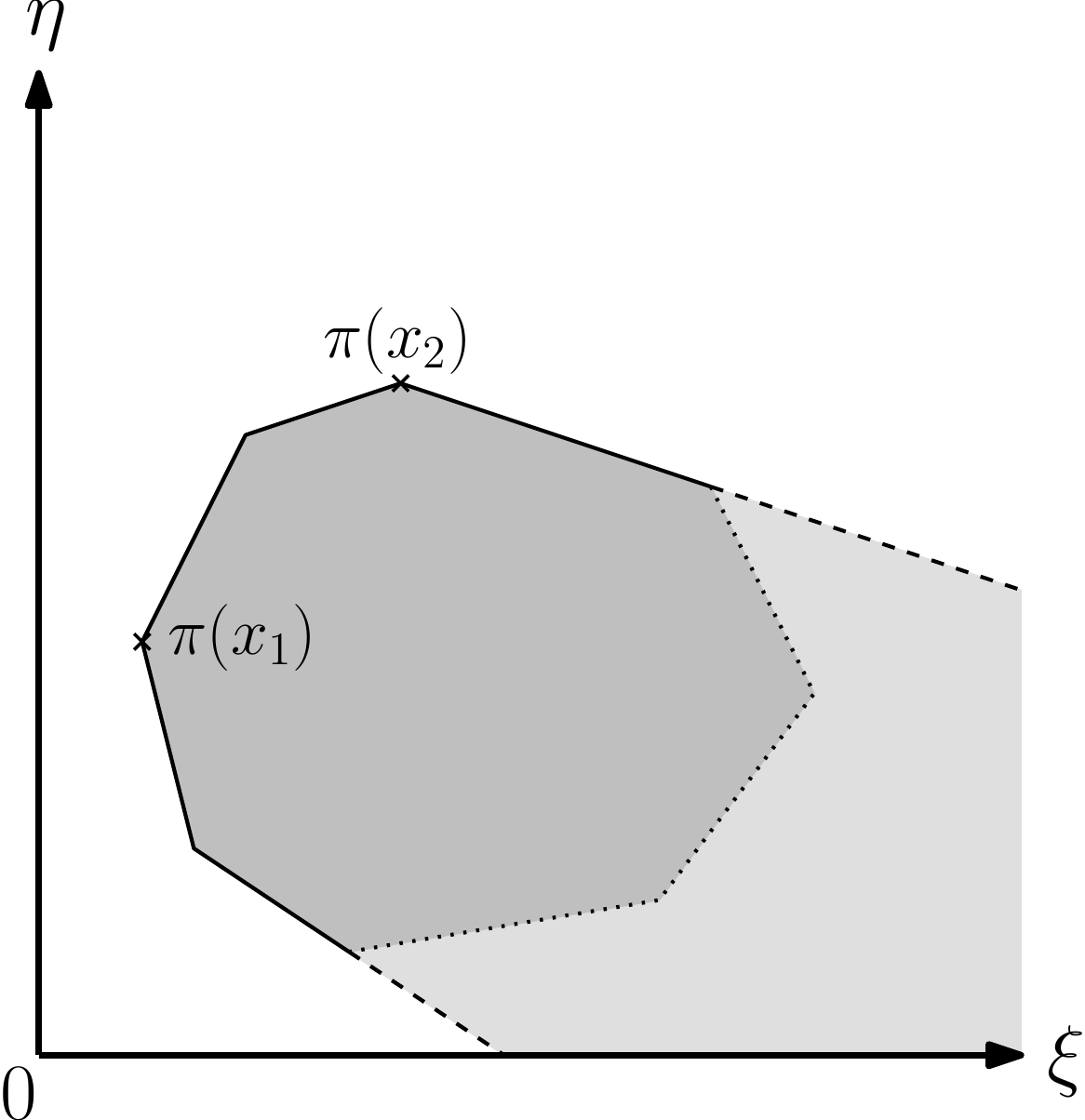}
  \end{center}

  \caption{Shadow polygon~$P'$}
  \label{fig:shadow vertex}
\end{figure}

\section{Degeneracy}
\label{sec:degeneracy}

Any degenerate polytope~$P$ can be made non-degenerate by
perturbing the vector~$b$ by a tiny amount of random noise. This way, another 
polytope~$\tilde{P}$ is obtained that is non-degenerate with probability one. Any degenerate
vertex of~$P$ at which~$\ell>n$ constraints are tight generates at most~$\binom{\ell}{n}$ vertices of~$\tilde{P}$
that are all very close to each other if the perturbation of~$b$ is small. 
We say that two vertices of~$\tilde{P}$
that correspond to the same vertex of~$P$ are in the same equivalence class.

If the perturbation of the vector~$b$ is small enough, then any edge between two vertices of~$\tilde{P}$
in different equivalence classes corresponds to an edge in~$P$ between the vertices that generated these
equivalence classes.
We apply the shadow vertex algorithm to the polytope~$\tilde{P}$ to find a path~$R$ between two arbitrary vertices
from the equivalence classes generated by~$x_1$ and~$x_2$, respectively.
Then we translate this path into a walk from~$x_1$ to~$x_2$ on the polytope~$P$ by mapping each
vertex on the path~$R$ to the vertex that generated its equivalence class. This way, we obtain
a walk from~$x_1$ to~$x_2$ on the polytope~$P$ that may visit vertices multiple times and may also stay
in the same vertex for some steps. In the latter type of steps only the algebraic representation of the current vertex
is changed. As this walk on~$P$ has the same length as the path that the shadow vertex
algorithm computes on~$\tilde{P}$, the upper bound we derive for the length of~$R$ also applies to
the degenerate polytope~$P$.  

Of course the perturbation of the vector~$b$ might change the shape of the
polytope~$P$. In this context it is important to point out that the parameter~$\delta$, which we define in the
following, only depends on the matrix~$A$ and, thus, is independent of the
right-hand side~$b$. Consequently, the parameter~$\delta$ of the original
polytope~$P$ can also be used to describe the behavior of the shadow vertex simplex
algorithm on the polytope~$\tilde{P}$.

\section{Outline of the Analysis}
\label{sec:Outline}

In the remainder of this paper we assume that the polytope~$P$ is
non-degenerate, i.e., for each vertex~$x$ of~$P$ there are exactly~$n$
indices~$i$ for which $a_i^\T x = b_i$. This implies that for any edge between two vertices~$x$ and~$y$ of~$P$
there are exactly~$n-1$ indices~$i$ for which~$a_i^\T x = a_i^\T y = b_i$. According to Section~\ref{sec:degeneracy}
this assumption is justified.

From the description of the shadow vertex algorithm it is clear that the main step in proving
Theorem~\ref{maintheorem} is to bound the expected number of edges on the path from~$\pi(x_1)$ to~$\pi(x_2)$
on the polygon~$P'$. In order to do this, we look at the slopes of the edges on this path.
As we discussed above, the sequence of slopes is monotonically decreasing. We will show that due to
the randomness in the objective functions~$w_1$ and~$w_2$, it is even strictly decreasing with probability one.
Furthermore all slopes on this path are bounded from below by~0.

Instead of counting the edges on the path from~$\pi(x_1)$ to~$\pi(x_2)$ directly, 
we will count the number of different slopes in the interval~$[0,1]$ and
we observe that the expected number of slopes from the interval $[0,\infty)$ is
twice the expected number of slopes from the interval~$[0,1]$.
In order to count the number of slopes in~$[0,1]$,
we partition the interval~$[0,1]$ into several small subintervals and we bound for each of these subintervals~$I$ the expected number of slopes in~$I$. Then
we use linearity of expectation to obtain an upper bound on the expected number of different slopes
in~$[0,1]$, which directly translates into an upper bound on the expected number of edges on
the path from~$\pi(x_1)$ to~$\pi(x_2)$.

We choose the subintervals so small that, with high probability, none of them contains more than one slope.
Then, the expected number of slopes in a subinterval~$I=(t,t+\e]$ is approximately equal to
the probability that there is a slope in the interval~$I$.
In order to bound this probability, we use a technique reminiscent of the principle of deferred decisions
that we have already used in~\cite{BrunschR12}. The main idea is to split the random draw
of the vectors~$w_1$ and~$w_2$ in the shadow vertex algorithm into two steps.
The first step reveals enough information about the realizations of these vectors
to determine the last edge~$e=(\hat{p},p^{\star})$ on the path from~$\pi(x_1)$ to~$\pi(x_2)$
whose slope is bigger than~$t$ (see Figure~\ref{fig:slopes}). Even though~$e$ is determined in the
first step, its slope is not. 
We argue that there is still enough randomness left in
the second step to bound the probability that the slope of~$e$ lies in the interval~$(t,t+\e]$ from above,
yielding Theorem~\ref{maintheorem}.

We will now give some more details on how the random draw of the vectors~$w_1$ and~$w_2$ is partitioned.
Let~$\hat{x}$ and~$x^{\star}$ be the vertices of the polytope~$P$ that are projected onto~$\hat{p}$
and~$p^{\star}$, respectively. Due to the non-degeneracy of the polytope~$P$, there are exactly~$n-1$
constraints that are tight for both~$\hat{x}$ and~$x^{\star}$ and there is a unique constraint~$a_i^\T x\le b_i$
that is tight for~$x^{\star}$ but not for~$\hat{x}$. In the first step the vector~$w_1$ is completely
revealed while instead of~$w_2$ only an element~$\tilde{w}_2$ from the ray~$\SET{ w_2 + \gamma \cdot a_i \WHERE \gamma \geq 0}$
is revealed. We then argue that knowing~$w_1$ and~$\tilde{w}_2$ suffices to identify the edge~$e$.
The only randomness left in the second step is the exact position of the vector~$w_2$ on the
ray~$\SET{ \tilde{w}_2 - \gamma \cdot a_i \WHERE \gamma \geq 0}$, which suffices to bound
the probability that the slope of~$e$ lies in the interval $(t,t+\e]$.

Let us remark that the proof of Theorem~\ref{maintheorem} is inspired by the recent smoothed analysis of the successive
shortest path algorithm for the minimum-cost flow problem~\cite{BrunschCMR13}.  
Even though the general structure bears some similarity, the details of our analysis are
much more involved.

\section{The Parameter~\texorpdfstring{\boldmath$\delta$\unboldmath}{delta}}
\label{sec:Parameter}

In this section we define the parameter~$\delta$ that describes the flatness of the vertices of the polytope and state some relevant properties.

\begin{definition}\noindent
\begin{enumerate}

  \item Let $z_1, \ldots, z_n \in \RR^n$ be linearly independent vectors and let $\varphi \in (0, \frac{\pi}{2}]$ be the angle between~$z_n$ and the hyperplane $\SPAN{ z_1, \ldots, z_{n-1} }$. By $\hat{\delta}(\SET{ z_1, \ldots, z_{n-1} }, z_n) = \sin \varphi$ we denote the sine of angle~$\varphi$. Moreover, we set $\delta(z_1, \ldots, z_n) = \min_{k \in [n]} \hat{\delta}(\SET{ z_i \WHERE i \in [n] \setminus \SET{k} }, z_k)$.

  \item Given a matrix $A = [a_1, \ldots, a_m]^\T \in \RR^{m \times n}$, we set
  \vspace{-0.5em}\[
    \delta(A) = \min \SET{ \delta(a_{i_1}, \ldots, a_{i_n}) \WHERE a_{i_1}, \ldots, a_{i_n}\ \text{linearly independent} } \DOT \vspace{-0.5em}
  \]

\end{enumerate}
\end{definition}

The value $\hat{\delta}(\SET{ z_1, \ldots, z_{n-1} }, z_n)$ describes how orthogonal~$z_n$ is to the span of $z_1, \ldots, z_{n-1}$. If $\varphi \approx 0$, i.e., $z_n$ is close to the span of $z_1, \ldots, z_{n-1}$, then $\hat{\delta}(\SET{ z_1, \ldots, z_{n-1} }, z_n) \approx 0$. On the other hand, if~$z_n$ is orthogonal to $z_1, \ldots, z_{n-1}$, then $\varphi = \frac{\pi}{2}$ and, hence, $\hat{\delta}(\SET{ z_1, \ldots, z_{n-1} }, z_n) = 1$. The value $\hat{\delta}(\SET{ z_1, \ldots, z_{n-1} }, z_n)$ equals the distance between both faces of the parallelotope~$Q$, given by $Q = \SET{ \sum_{i=1}^n \alpha_i \cdot \N(z_i) \WHERE \alpha_i \in [0, 1] }$, that are parallel to $\SPAN{ z_1, \ldots, z_{n-1} }$ and is scale invariant.

The value $\delta(z_1, \ldots, z_n)$ equals twice the inner radius~$r_n$ of the parallelotope~$Q$ and, thus, is a measure of the flatness of~$Q$: A value $\delta(z_1, \ldots, z_n) \approx 0$ implies that~$Q$ is nearly $(n-1)$-dimensional. On the other hand, if $\delta(z_1, \ldots, z_n) = 1$, then the vectors $z_1, \ldots, z_n$ are pairwise orthogonal, that is, $Q$ is an $n$-dimensional unit cube.

The next lemma lists some useful statements concerning the parameter $\delta \DEF \delta(A)$ including a connection to the parameters~$\Delta_1$, $\Delta_{n-1}$, and~$\Delta$ introduced in the paper of Bonifas et al.~\cite{BonifasDEHN12}.

\begin{lemma}
\label{lemma:delta properties}
Let $z_1, \ldots, z_n \in \RR^n$ be linearly independent vectors, let $A \in \RR^{m \times n}$ be a matrix, let $b \in \RR^m$ be a vector, and let $\delta = \delta(A)$. Then, the following claims hold true:
\begin{enumerate}

  \item \label{delta properties:inverse} If~$M$ is the inverse of $[\N(z_1), \ldots, \N(z_n)]^\T$, then
  \vspace{-0.5em}\[
    \delta(z_1, \ldots, z_n)
    = \frac{1}{\max_{k \in [n]} \|m_k\|}
    \leq \frac{\sqrt{n}}{\max_{k \in [n]} \|M_k\|} \COMMA \vspace{-0.5em}
  \]
  where $[m_1, \ldots, m_n] = M$ and $[M_1, \ldots, M_n] = M^\T$.
  
  \item \label{delta properties:orthogonal matrix} If $Q \in \RR^{n \times n}$ is an orthogonal matrix, then $\delta(Qz_1, \ldots, Qz_n) = \delta(z_1, \ldots, z_n)$.

  \item \label{delta properties:neighboring vertices} Let~$y_1$ and~$y_2$ be two neighboring vertices of $P = \SET{ x \in \RR^n \WHERE Ax \leq b }$ and let~$a_i^\T$ be a row of~$A$. If $a_i^\T \cdot (y_2-y_1) \neq 0$, then $|a_i^\T \cdot (y_2-y_1)| \geq \delta \cdot \|y_2-y_1\|$.

  \item \label{delta properties:comparison} If~$A$ is an integral matrix, then $\frac{1}{\delta} \leq n \Delta_1 \Delta_{n-1} \leq n \Delta^2$, where~$\Delta$, $\Delta_1$, and $\Delta_{n-1}$ are the largest absolute values of any sub-determinant of~$A$ of arbitrary size, of size~$1$, and of size~$n-1$, respectively.
  
\end{enumerate}
\end{lemma}

\begin{proof}
First of all we derive a simple formula for $\hat{\delta}(\SET{ z_1, \ldots, z_{n-1} }, z_n)$. For this, assume that the vectors $z_1, \ldots, z_n$ are normalized. Now consider a normal vector $x \neq \NULL$ of $\SPAN{z_1, \ldots, z_{n-1}}$ that lies in the same halfspace as~$z_n$. Let $\varphi \in (0, \frac{\pi}{2}]$ be the angle between~$z_n$ and $\SPAN{z_1, \ldots, z_{n-1}}$ and let $\psi \in [0, \frac{\pi}{2})$ be the angle between~$z_n$ and~$x$. Clearly, $\varphi + \psi = \frac{\pi}{2}$. Consequently,
\[
  \hat{\delta}(\SET{ z_1, \ldots, z_{n-1} }, z_n)
  = \sin \varphi
  = \sin \left( \frac{\pi}{2} - \psi \right)
  = \cos \psi
  = \frac{z_n^\T x}{\|x\|} \DOT
\]
The last fraction is invariant under scaling of~$x$. Since~$x$ and~$z_n$ lie in the same halfspace, w.l.o.g.\ we can assume that $z_n^\T x = 1$. Hence, $\hat{\delta}(\SET{ z_1, \ldots, z_{n-1} }, z_n) = \frac{1}{\|x\|}$, where~$x$ is the unique solution of the equation $[z_1, \ldots, z_{n-1}, z_n]^\T \cdot x = (0, \ldots, 0, 1)^\T = e_n$. If the vectors $z_1, \ldots, z_n$ are not normalized, then we obtain
\[
  \hat{\delta}(\SET{ z_1, \ldots, z_{n-1} }, z_n)
  = \hat{\delta}(\SET{ \N(z_1), \ldots, \N(z_{n-1}) }, \N(z_n))
  = \frac{1}{\|x\|} \COMMA
\]
where $x = [\N(z_1), \ldots, \N(z_{n-1}), \N(z_n)]^{-\T} \cdot e_n$. Since for the previous line of reasoning we can relabel the vectors $z_1, \ldots, z_n$ arbitrarily, this implies
\begin{align*}
  \delta(z_1, \ldots, z_n)
  &= \min_{k \in [n]} \frac{1}{\left\| [\N(z_1), \ldots, \N(z_n)]^{-\T} \cdot e_k \right\|} \cr
  &= \frac{1}{\max \SET{ \|x\| \WHERE \text{$x$ is column of $[\N(z_1), \ldots, \N(z_n)]^{-\T}$} }} \DOT
\end{align*}
This yields the equation in Claim~\ref{delta properties:inverse}. Due to
\[
  \left( \max_{k \in [n]} \|M_k\| \right)^2
  \leq \sum_{k \in [n]} \|M_k\|^2
  = \sum_{k \in [n]} \|m_k\|^2
  \leq n \cdot \left( \max_{k \in [n]} \|m_k\| \right)^2
\]
we obtain the inequality $\frac{1}{\max_{k \in [n]} \|m_k\|} \leq \frac{\sqrt{n}}{\max_{k \in [n]} \|M_k\|}$ stated in Claim~\ref{delta properties:inverse}.

For Claim~\ref{delta properties:orthogonal matrix} observe that
\begin{align*}
  [\N(Q z_1), \ldots, \N(Q z_n)]^{-\T}
  &= [Q \N(z_1), \ldots, Q \N(z_n)]^{-\T} \cr
  &= (Q \cdot [\N(z_1), \ldots, \N(z_n)])^{-\T} \cr
  &= ([\N(z_1), \ldots, \N(z_n)]^{-1} \cdot Q^\T)^\T \cr
  &= Q \cdot [\N(z_1), \ldots, \N(z_n)]^{-\T}
\end{align*}
for any orthogonal matrix~$Q$. Therefore, we get
\begin{align*}
  \frac{1}{\delta(Qz_1, \ldots, Qz_n)}
  &= \max \SET{ \|x\| \WHERE \text{$x$ is column of $[\N(Qz_1), \ldots, \N(Qz_n)]^{-\T}$} } \cr
  &= \max \SET{ \|Qy\| \WHERE \text{$y$ is column of $[\N(z_1), \ldots, \N(z_n)]^{-\T}$} } \cr
  &= \max \SET{ \|y\| \WHERE \text{$y$ is column of $[\N(z_1), \ldots, \N(z_n)]^{-\T}$} } \cr
  &= \frac{1}{\delta(z_1, \ldots, z_n)} \DOT
\end{align*}
For Claim~\ref{delta properties:neighboring vertices} let~$y_1$ and~$y_2$ be two neighboring vertices of~$P$. Then, there are exactly~$n-1$ indices~$j$ for which $a_j^\T \cdot (y_2-y_1) = 0$. We denote them by $j_1, \ldots, j_{n-1}$. If there is an index~$i$ for which $a_i^\T \cdot (y_2-y_1) \neq 0$, then $a_{j_1}, \ldots, a_{j_{n-1}}, a_i$ are linearly independent. Consequently, $\delta(a_{j_1}, \ldots, a_{j_{n-1}}, a_i) \geq \delta$. Let us assume that $a_i^\T \cdot (y_2-y_1) \geq 0$. (Otherwise, consider $a_i^\T \cdot (y_1-y_2)$ instead.) Since $y_2-y_1$ is a normal vector of $\SPAN{a_{j_1}, \ldots, a_{j_{n-1}}}$ that lies in the same halfspace as~$a_i$, we obtain
\[
  \frac{a_i^\T \cdot (y_2-y_1)}{\|y_2-y_1\|}
  = \hat{\delta}(\SET{ a_{j_1}, \ldots, a_{j_{n-1}} }, a_i)
  \geq \delta(a_{j_1}, \ldots, a_{j_{n-1}}, a_i)
  \geq \delta
\]
and, thus, $a_i^\T \cdot (y_2-y_1) \geq \delta \cdot \|y_2-y_1\|$.

For proving Claim~\ref{delta properties:comparison} we can focus on showing the first inequality. The second one follows from $\Delta \geq \max \SET{ \Delta_1, \Delta_{n-1} }$. For this, it suffices to show that for~$n$ arbitrary linearly independent rows $a_{i_1}^\T, \ldots, a_{i_n}^\T$ of~$A$ the inequality
\[
  \frac{1}{\hat{\delta}(\SET{ a_{i_1}, \ldots, a_{i_{n-1}} }, a_{i_n})}
  \leq n \Delta_1 \Delta_{n-1}
\]
holds. By previous observations we know that
\[
  \frac{1}{\hat{\delta}(\SET{ a_{i_1}, \ldots, a_{i_{n-1}} }, a_{i_n})}
    = \|x\|
\]
where~$x$ is the unique solution of $\hat{A} x = e_n$ for $\hat{A} = [\N(a_{i_1}), \ldots, \N(a_{i_n})]^\T$. Let $\tilde{A} = [a_{i_1}, \ldots, a_{i_n}]^\T$. Then,
\begin{align*}
  \|x\|^2
  &= \sum_{k=1}^n x_k^2
  = \sum_{k=1}^n \left( \frac{\det(\hat{A}_{n,k})}{\det(\hat{A})} \right)^2
  = \sum_{k=1}^n \left( \frac{\det(\tilde{A}_{n,k}) \cdot \prod_{j=1}^{n-1} \frac{1}{\|a_{i_j}\|}}{\det(\tilde{A}) \cdot \prod_{j=1}^n \frac{1}{\|a_{i_j}\|}} \right)^2 \cr
  &= \sum_{k=1}^n \left( \frac{\det(\tilde{A}_{n,k}) \cdot \|a_{i_n}\|}{\det(\tilde{A})} \right)^2
  \leq \sum_{k=1}^n \left( \frac{\Delta_{n-1} \cdot \sqrt{n} \Delta_1}{1} \right)^2
  = n^2 \Delta_1^2 \Delta_{n-1}^2 \DOT
\end{align*}
Some of the equations need further explanation: Due to Cramer's rule, we have $x_k = \frac{\det(\bar{A})}{\det(\hat{A})}$, where~$\bar{A}$ is obtained from~$\hat{A}$ by replacing the $k ^\text{th}$ column by the right-hand side~$e_n$ of the equation $\hat{A} x = e_n$. Laplace's formula yields $|\det(\bar{A})| = |\det(\hat{A}_{n,k})|$. Hence, the second equation is true. For the third equation note that the $k^\text{th}$ row of matrix~$\hat{A}$ is the same as the $k^\text{th}$ row of matrix~$\tilde{A}$ up to a factor of $\frac{1}{\|a_{i_k}\|}$. The inequality follows from $|\det(\tilde{A}_{n,k})| \leq \Delta_{n-1}$ since this is a sub-determinant of~$A$ of size~$n-1$, from $\|a_{i_n}\| \leq \sqrt{n} \cdot \|a_{i_n}\|_\infty \leq \sqrt{n} \Delta_1$, since $\|a_{i_n}\|_\infty$ is a sub-determinant of~$A$ of size~$1$, and from $|\det(\tilde{A})| \geq 1$ since~$\tilde{A}$ is invertible and integral by assumption. Hence,
\[
  \frac{1}{\hat{\delta}(\SET{ a_{i_1}, \ldots, a_{i_{n-1}} }, a_{i_n})}
  = \|x\|
  \leq n \Delta_1 \Delta_{n-1} \DOT \qedhere
\]
\end{proof}

\section{Analysis}
\label{sec:analysis}

For the proof of Theorem~\ref{maintheorem} we assume that $\|a_i\| = 1$ for all $i \in [m]$. This entails no loss of generality since normalizing the rows of matrix~$A$ (and scaling the right-hand side~$b$ appropriately) does neither change the behavior of our algorithm nor does it change the parameter $\delta = \delta(A)$.

For given linear functions~$L_1$ and~$L_2$, we denote by $\pi = \pi_{L_1, L_2}$ the function $\pi \colon \RR^n \to \RR^2$, given by $\pi(x) = (L_1(x), L_2(x))$. Note, that $n$-dimensional vectors can be treated as linear functions. By $P' = P'_{L_1, L_2}$ we denote the projection $\pi(P)$ of polytope~$P$ onto the Euclidean plane, and by $R = R_{L_1, L_2}$ we denote the path from~$\pi(x_1)$ to~$\pi(x_2)$ along the edges of polygon~$P'$.

Our goal is to bound the expected number of edges of the path $R = R_{w_1, w_2}$ which is random since~$w_1$ and~$w_2$ depend on the realizations of the random vectors~$\lambda$ and~$\mu$. Each edge of~$R$ corresponds to a slope in $(0, \infty)$. These slopes are pairwise distinct with probability one (see Lemma~\ref{lemma:failure probability II}). Hence, the number of edges of~$R$ equals the number of distinct slopes of~$R$. In order to bound the expected number of distinct slopes we first restrict our attention to slopes in the interval $(0, 1]$.

\begin{definition}
\label{definition:failure event}
For a real $\e > 0$ let~$\F_\e$ denote the event that there are three pairwise distinct vertices $z_1, z_2, z_3$ of~$P$ such that~$z_1$ and~$z_3$ are neighbors of~$z_2$ and such that
\[
  \left| \frac{w_2^\T \cdot (z_2-z_1)}{w_1^\T \cdot (z_2-z_1)} - \frac{w_2^\T \cdot (z_3-z_2)}{w_1^\T \cdot (z_3-z_2)} \right| \leq \e \DOT
\]
\end{definition}

Note that if event~$\F_\e$ does not occur, then all slopes of~$R$ differ by more than~$\e$. Particularly, all slopes are pairwise distinct. First of all we show that event~$\F_\e$ is very unlikely to occur if~$\e$ is chosen sufficiently small.

\begin{lemma}
\label{lemma:failure probability I}
The probability that there are two neighboring vertices $z_1, z_2$ of~$P$ such that $|w_1^\T \cdot (z_2-z_1)| \leq \e \cdot \|z_2-z_1\|$ is bounded from above by $\frac{2m^n\e}{\delta}$.
\end{lemma}

\begin{proof}
Let~$z_1$ and~$z_2$ be two neighbors of~$P$. Let $\Delta_z = z_2-z_1$. Because the claim we want to show is invariant under scaling, we can assume without loss of generality that $\|\Delta_z\| = 1$. There are $n-1$ indices $i_1, \ldots, i_{n-1} \in [m]$ such that $a_{i_k}^\T z_1 = b_{i_k} = a_{i_k}^\T z_2$. Recall that $w_1 = -[u_1, \ldots, u_n] \cdot \lambda$, where $\lambda = (\lambda_1, \ldots, \lambda_n)$ is drawn uniformly at random from $(0, 1]^n$. There must be an index~$i$ such that $a_{i_1}, \ldots, a_{i_{n-1}}, u_i$ are linearly independent. Hence, $\kappa \DEF u_i^\T \Delta_z \neq 0$ and, thus, $|\kappa| \geq \delta$ due to Lemma~\ref{lemma:delta properties}, Claim~\ref{delta properties:neighboring vertices}.

We apply the principle of deferred decisions and assume that all~$\lambda_j$ for $j \neq i$ are already drawn. Then
\[
  w_1^\T \Delta_z
  = -\sum_{j=1}^n \lambda_j \cdot u_j^\T \Delta_z
  = \underbrace{-\sum_{j \neq i} \lambda_j \cdot u_j^\T \Delta_z}_{\FED \gamma} - \lambda_i \cdot \kappa \DOT
\]
Thus,
\begin{align*}
  |w_1^\T \Delta_z| \leq \e
  &\iff w_1^\T \Delta_z \in [-\e, \e]
  \iff \lambda_i \cdot \kappa \in [\gamma - \e, \gamma + \e] \cr
  &\iff \lambda_i \in \left[ \frac{\gamma}{\kappa} - \frac{\e}{|\kappa|}, \frac{\gamma}{\kappa} + \frac{\e}{|\kappa|} \right] \DOT
\end{align*}
The probability for the latter event is bounded by the length of the interval, i.e., by $\frac{2\e}{|\kappa|} \leq \frac{2\e}{\delta}$. Since we have to consider at most $\binom{m}{n-1} \leq m^n$ pairs of neighbors $(z_1, z_2)$, applying a union bound yields the additional factor of~$m^n$.
\end{proof}

\begin{lemma}
\label{lemma:failure probability II}
The probability of event~$\F_\e$ tends to~$0$ for $\e \to 0$.
\end{lemma}

\begin{proof}
Let $z_1, z_2, z_3$ be pairwise distinct vertices of~$P$ such that~$z_1$ and~$z_3$ are neighbors of~$z_2$ and let $\Delta_z \DEF z_2-z_1$ and  $\Delta'_z \DEF z_3-z_2$. We assume that $\|\Delta_z\| = \|\Delta'_z\| = 1$. This entails no loss of generality as the fractions in Definition~\ref{definition:failure event} are invariant under scaling. Let $i_1, \ldots, i_{n-1} \in [m]$ be the $n-1$ indices for which $a_{i_k}^\T z_1 = b_{i_k} = a_{i_k}^\T z_2$. The rows $a_{i_1}, \ldots, a_{i_{n-1}}$ are linearly independent because~$P$ is non-degenerate. Since $z_1, z_2, z_3$ are distinct vertices of~$P$ and since~$z_1$ and~$z_3$ are neighbors of~$z_2$, there is exactly one index~$i_\ell$ for which $a_{i_\ell}^\T z_3 < b_{i_\ell}$, i.e., $a_{i_\ell}^\T \Delta'_z \neq 0$. Otherwise, $z_1, z_2, z_3$ would be collinear which would contradict the fact that they are distinct vertices of~$P$. Without loss of generality assume that $\ell = n-1$. Since $a_{i_k}^\T \Delta_z = 0$ for each $k \in [n-1]$, the vectors $a_{i_1}, \ldots, a_{i_{n-1}}, \Delta_z$ are linearly independent.

We apply the principle of deferred decisions and assume that~$w_1$ is already fixed. Thus, $w_1^\T \Delta_z$ and $w_1^\T \Delta'_z$ are fixed as well. Moreover, we assume that $w_1^\T \Delta_z \neq 0$ and $w_1^\T \Delta'_z \neq 0$ since this happens almost surely due to Lemma~\ref{lemma:failure probability I}. Now consider the matrix $M = [a_{i_1}, \ldots, a_{i_{n-2}}, \Delta_z, a_{i_{n-1}}]$ and the random vector
$
  (Y_1, \ldots, Y_{n-1}, Z)^\T
  = M^{-1} \cdot w_2
  = M^{-1} \cdot [v_1, \ldots, v_n] \cdot \mu
$.
For fixed values $y_1, \ldots, y_{n-1}$ let us consider all realizations of~$\mu$ for which $(Y_1, \ldots, Y_{n-1}) = (y_1, \ldots, y_{n-1})$. Then
\begin{align*}
  w_2^\T \Delta_z
  &= \big( M \cdot (y_1, \ldots, y_{n-1}, Z)^\T \big)^\T \Delta_z \cr
  &= \sum_{k=1}^{n-2} y_k \cdot a_{i_k}^\T \Delta_z + y_{n-1} \cdot \Delta_z^\T \Delta_z + Z \cdot a_{i_{n-1}}^\T \Delta_z \cr
  &= y_{n-1} \COMMA
\end{align*}
i.e., the value of $w_2^\T \Delta_z$ does not depend on the outcome of~$Z$ since~$\Delta_z$ is orthogonal to all~$a_{i_k}$. For~$\Delta'_z$ we obtain
\begin{align*}
  w_2^\T \Delta'_z
  &= \big( M \cdot (y_1, \ldots, y_{n-1}, Z)^\T \big)^\T \Delta'_z \cr
  &= \sum_{k=1}^{n-2} y_k \cdot a_{i_k}^\T \Delta'_z + y_{n-1} \cdot \Delta_z^\T \Delta'_z + Z \cdot a_{i_{n-1}}^\T \Delta'_z \cr
  &= \underbrace{y_{n-1} \cdot \Delta_z^\T \Delta'_z}_{\FED \kappa} + Z \cdot a_{i_{n-1}}^\T \Delta'_z
\end{align*}
as~$\Delta'_z$ is orthogonal to all~$a_{i_k}$ except for $k = \ell = n-1$. The chain of equivalences
\begin{align*}
  &\left| \frac{w_2^\T \Delta_z}{w_1^\T \Delta_z} - \frac{w_2^\T \Delta'_z}{w_1^\T \Delta'_z} \right| \leq \e
  \iff \frac{w_2^\T \Delta'_z}{w_1^\T \Delta'_z} \in \left[ \frac{w_2^\T \Delta_z}{w_1^\T \Delta_z} - \e, \frac{w_2^\T \Delta_z}{w_1^\T \Delta_z} + \e \right] \cr
  &\iff w_2^\T \Delta'_z \in \left[ \frac{w_2^\T \Delta_z}{w_1^\T \Delta_z} \cdot w_1^\T \Delta'_z - \e \cdot |w_1^\T \Delta'_z|, \frac{w_2^\T \Delta_z}{w_1^\T \Delta_z} \cdot w_1^\T \Delta'_z + \e \cdot |w_1^\T \Delta'_z| \right] \cr
  &\iff Z \cdot a_{i_{n-1}}^\T \Delta'_z \in \left[ \frac{w_2^\T \Delta_z}{w_1^\T \Delta_z} \cdot w_1^\T \Delta'_z - \kappa - \e \cdot |w_1^\T \Delta'_z|, \frac{w_2^\T \Delta_z}{w_1^\T \Delta_z} \cdot w_1^\T \Delta'_z - \kappa + \e \cdot |w_1^\T \Delta'_z| \right]
\end{align*}
implies, that for event~$\F_\e$ to occur~$Z$ must fall into an interval $I = I(y_1, \ldots, y_{n-1})$ of length $2\e \cdot \frac{|w_1^\T \Delta'_z|}{|a_{i_{n-1}}^\T \Delta'_z|}$. The probability of this is bounded from above by
\[
  \frac{2n \cdot 2\e \cdot \frac{|w_1^\T \Delta'_z|}{|a_{i_{n-1}}^\T \Delta'_z|}}{\delta(r_1, \ldots, r_n) \cdot \min_{k \in [n]} \|r_k\|}
  = \underbrace{\frac{4n \cdot |w_1^\T \Delta'_z|}{\delta(r_1, \ldots, r_n) \cdot \min_{k \in [n]} \|r_k\| \cdot |a_{i_{n-1}}^\T \Delta'_z|}}_{\FED \gamma} \cdot \e \COMMA
\]
where $[r_1, \ldots, r_n] = M^{-1} \cdot [v_1, \ldots, v_n]$. This is due to $(Y_1, \ldots, Y_{n-1}, Z)^\T = [r_1, \ldots, r_n] \cdot \mu$ and Theorem~\ref{theorem.Prob:enough randomness}. Since the vectors $r_1, \ldots, r_n$ are linearly independent, we have $\delta(r_1, \ldots, r_n) > 0$ and $\min_{k \in [n]} \|r_k\| > 0$. Furthermore, $|a_{i_{n-1}}^\T \Delta'_z| > 0$ since~$i_{n-1}$ is the constraint which is not tight for~$z_3$, but for~$z_2$. Hence, $\gamma < \infty$, and thus $\Pr{\left| \frac{w_2^\T \Delta_z}{w_1^\T \Delta_z} - \frac{w_2^\T \Delta'_z}{w_1^\T \Delta'_z} \right| \leq \e} \to 0$ for $\e \to 0$.

As there are at most $m^{3n}$ triples $(z_1, z_2, z_3)$ we have to consider, the claim follows by applying a union bound.
\end{proof}

Let $p \neq \pi(x_2)$ be a vertex of~$R$. We call the slope~$s$ of the edge incident to~$p$ to the right of~$p$ \emph{the slope of~$p$}. As a convention, we set the slope of $\pi(x_2)$ to~$0$ which is smaller than the slope of any other vertex~$p$ of~$R$.

\begin{figure}
  \begin{center}
    \includegraphics[width=0.3\textwidth]{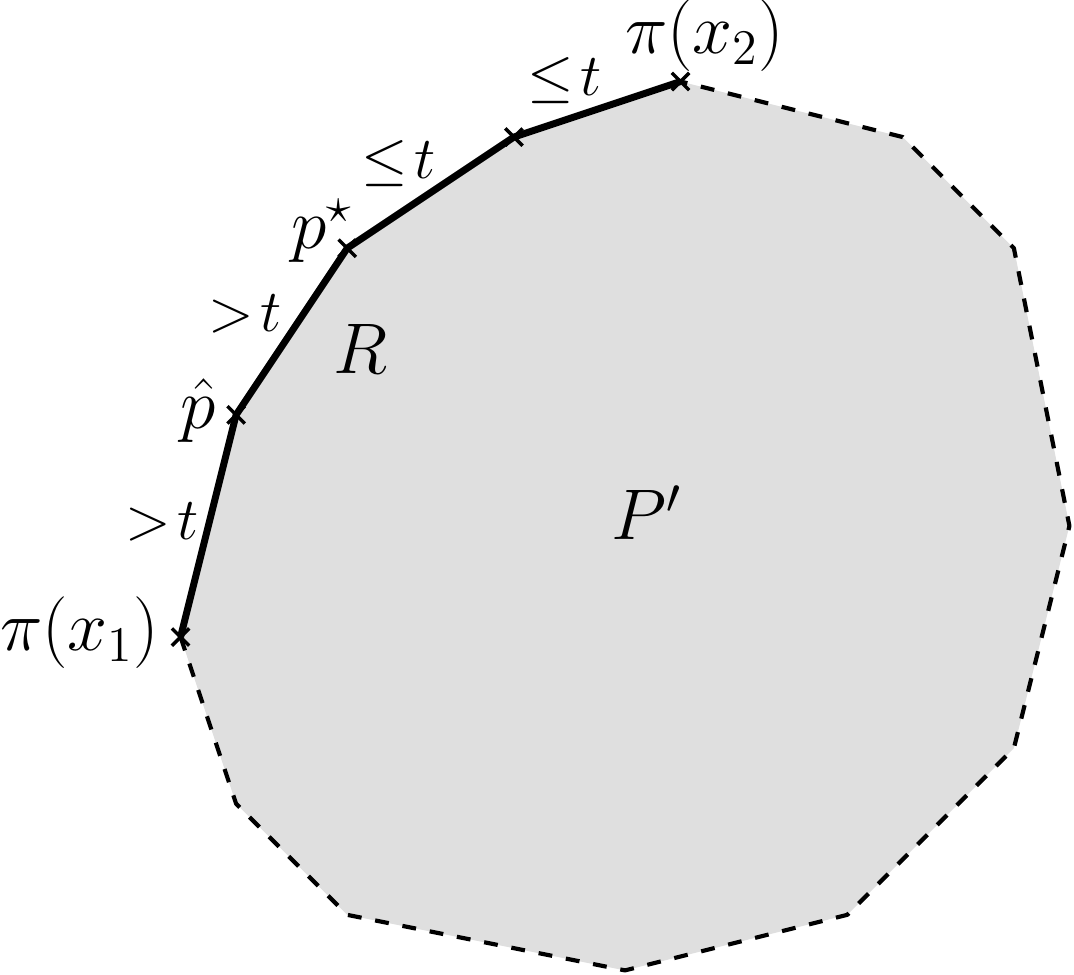}
  \end{center}
  \caption{Slopes of the vertices of~$R$}
  \label{fig:slopes}
\end{figure}

Let $t \geq 0$ be an arbitrary real, let~$\hat{p}$ be the right-most vertex of~$R$ whose slope is larger than~$t$, and let~$p^\star$ be the right neighbor of~$\hat{p}$ (see Figure~\ref{fig:slopes}). Let~$\hat{x}$ and~$x^\star$ be the neighboring vertices of~$P$ with $\pi(\hat{x}) = \hat{p}$ and $\pi(x^\star) = p^\star$. Now let $i = i(x^\star, \hat{x}) \in [m]$ be the index for which $a_i^\T x^\star = b_i$ and for which~$\hat{x}$ is the (unique) neighbor~$x$ of~$x^\star$ for which $a_i^\T x < b_i$. This index is unique due to the non-degeneracy of the polytope~$P$. For an arbitrary real $\gamma \geq 0$ we consider the vector $\tilde{w}_2 = w_2 + \gamma \cdot a_i$.

\begin{lemma}
\label{lemma:reconstruct}
Let $\tilde{\pi} = \pi_{w_1, \tilde{w}_2}$ and let $\tilde{R} = R_{w_1, \tilde{w}_2}$ be the path from $\tilde{\pi}(x_1)$ to $\tilde{\pi}(x_2)$ in the projection $\tilde{P}' = P'_{w_1, \tilde{w}_2}$ of polytope~$P$. Furthermore, let~$\tilde{p}^\star$ be the left-most vertex of~$\tilde{R}$ whose slope does not exceed~$t$. Then,
$
  \tilde{p}^\star = \tilde{\pi}(x^\star)
$.
\end{lemma}

Let us reformulate the statement of Lemma~\ref{lemma:reconstruct} as follows: The vertex~$\tilde{p}^\star$ is defined for the path~$\tilde{R}$ of polygon~$\tilde{P}'$ with the same rules as used to define the vertex~$p^\star$ of the original path~$R$ of polygon~$P'$. Even though~$R$ and~$\tilde{R}$ can be very different in shape, both vertices, $p^\star$ and~$\tilde{p}^\star$, correspond to the same solution~$x^\star$ in the polytope~$P$, that is, $p^\star = \pi(x^\star)$ and $\tilde{p}^\star = \tilde{\pi}(x^\star)$. Let us remark that Lemma~\ref{lemma:reconstruct} is a significant generalization of Lemma~4.3 of~\cite{BrunschCMR13}.

\begin{proof}
We consider a linear auxiliary function $\bar{w}_2 \colon \RR^n \to \RR$, given by $\bar{w}_2(x) = \tilde{w}_2^\T x - \gamma \cdot b_i$. The paths $\bar{R} = R_{w_1, \bar{w}_2}$ and~$\tilde{R}$ are identical except for a shift by $-\gamma \cdot b_i$ in the second coordinate because for $\bar{\pi} = \pi_{w_1, \bar{w}_2}$ we obtain
\vspace{-0.5em}\[
  \bar{\pi}(x)
  = (w_1^\T x, \tilde{w}_2^\T x - \gamma \cdot b_i)
  = (w_1^\T x, \tilde{w}_2^\T x) - (0,  \gamma \cdot b_i)
  = \tilde{\pi}(x) - (0,  \gamma \cdot b_i) \vspace{-0.5em}
\]
for all $x \in \RR^n$. Consequently, the slopes of~$\bar{R}$ and~$\tilde{R}$ are exactly the same (see Figure~\ref{fig:reconstruct shift}).

\begin{figure}
  \begin{center}
    \begin{subfigure}{0.4\textwidth}
      \begin{center}
        \includegraphics[page=1, width=0.6\textwidth]{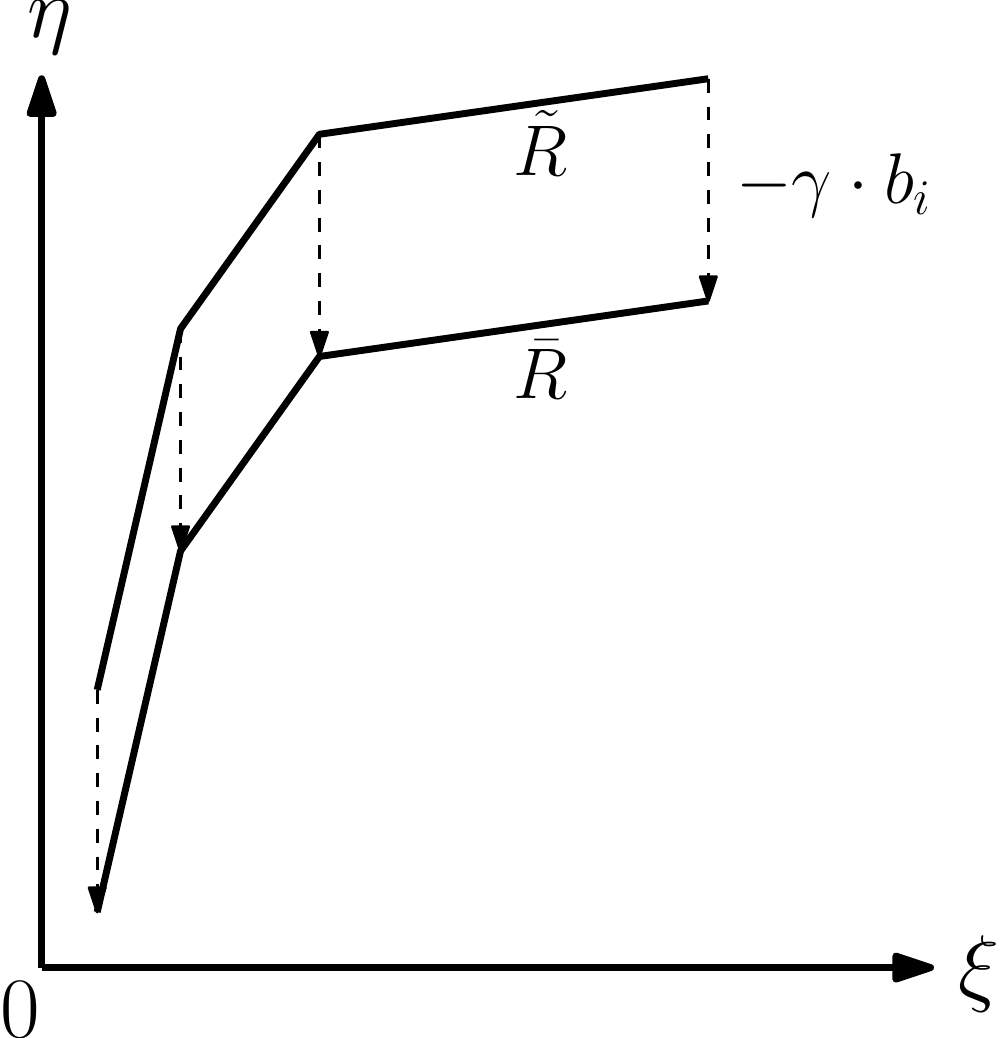}
        \caption{Relation between $\bar{R}$ and $\tilde{R}$}
        \label{fig:reconstruct shift}
      \end{center}
    \end{subfigure}
    \hspace{10ex}
    \begin{subfigure}{0.4\textwidth}
      \begin{center}
        \includegraphics[page=2, width=0.6\textwidth]{imgReconstruct.pdf}
        \caption{Relation between $\bar{R}$ an $R$}
        \label{fig:reconstruct below}
      \end{center}
    \end{subfigure}
  \end{center}
  \caption{Relations between $R$, $\tilde{R}$, and $\bar{R}$}
\end{figure}

Let $x \in P$ be an arbitrary point from the polytope~$P$. Then,
$
  \tilde{w}_2^\T x
  = w_2^\T x + \gamma \cdot a_i^\T x
  \leq w_2^\T x + \gamma \cdot b_i
$.
The inequality is due to $\gamma \geq 0$ and $a_i^\T x \leq b_i$ for all $x \in P$. Equality holds, among others, for $x = x^\star$ due to the choice of~$a_i$. Hence, for all points $x \in P$ the two-dimensional points $\pi(x)$ and $\bar{\pi}(x)$ agree in the first coordinate while the second coordinate of $\pi(x)$ is at least the second coordinate of $\bar{\pi}(x)$ as $\bar{w}_2(x) = \tilde{w}_2^\T x - \gamma \cdot b_i \leq w_2^\T x$. Additionally, we have $\pi(x^\star) = \bar{\pi}(x^\star)$. Thus, path~$\bar{R}$ is below path~$R$ but they meet at point $p^\star = \pi(x^\star)$. Hence, the slope of~$\bar{R}$ to the left (right) of~$p^\star$ is at least (at most) the slope of~$R$ to the left (right) of~$p^\star$ which is greater than (at most)~$t$ (see Figure~\ref{fig:reconstruct below}). Consequently, $p^\star$ is the left-most vertex of~$\bar{R}$ whose slope does not exceed~$t$. Since~$\bar{R}$ and~$\tilde{R}$ are identical up to a shift of $-(0, \gamma \cdot b_i)$, $\tilde{\pi}(x^\star)$ is the left-most vertex of~$\tilde{R}$ whose slope does not exceed~$t$, i.e., $\tilde{\pi}(x^\star) = \tilde{p}^\star$.
\end{proof}

Lemma~\ref{lemma:reconstruct} holds for any vector~$\tilde{w}_2$ on the ray $\vec{r} = \SET{ w_2 + \gamma \cdot a_i \WHERE \gamma \geq 0}$. As $\|w_2\| \leq n$ (see Section~\ref{sec:our algorithm}), we have $w_2 \in [-n,n]^n$. Hence, ray~$\vec{r}$ intersects the boundary of $[-n,n]^n$ in a unique point~$z$. We choose $\tilde{w}_2 = \tilde{w}_2(w_2, i) \DEF z$ and obtain the following result.

\begin{corollary}
\label{corollary:reconstruct}
Let $\tilde{\pi} = \pi_{w_1, \tilde{w}_2(w_2, i)}$ and let~$\tilde{p}^\star$ be the left-most vertex of path $\tilde{R} = R_{w_1, \tilde{w}_2(w_2, i)}$ whose slope does not exceed~$t$. Then,
$
	\tilde{p}^\star = \tilde{\pi}(x^\star)
$.
\end{corollary}

Note, that Corollary~\ref{corollary:reconstruct} only holds for the right choice of index $i = i(x^\star, \hat{x})$. The vector $\tilde{w}_2(w_2, i)$ is defined for any vector $w_2 \in [-n, n]^n$ and any index $i \in [m]$. In the remainder, index~$i$ is an arbitrary index from~$[m]$.

We can now define the following event that is parameterized in~$i$, $t$, and a real $\e > 0$ and that depends on~$w_1$ and~$w_2$.

\begin{definition}
\label{definition:event E}
For an index $i \in [m]$ and a real $t \geq 0$ let~$\tilde{p}^\star$ be the left-most vertex of $\tilde{R} = R_{w_1, \tilde{w}_2(w_2, i)}$ whose slope does not exceed~$t$ and let~$y^\star$ be the corresponding vertex of~$P$. For a real $\e > 0$ we denote by~$\E_{i, t, \e}$ the event that the conditions
\begin{itemize}

  \item[$\bullet$] $a_i^\T y^\star = b_i$ and

  \item[$\bullet$] $\frac{w_2^\T (\hat{y} - y^\star)}{w_1^\T (\hat{y} - y^\star)} \in (t, t+\e]$, where~$\hat{y}$ is the neighbor~$y$ of~$y^\star$ for which $a_i^\T y < b_i$,

\end{itemize}
are met. Note, that the vertex~$\hat{y}$ always exists and that it is unique since the polytope~$P$ is non-degenerate.
\end{definition}

Let us remark that the vertices~$y^\star$ and~$\hat{y}$, which depend on the index~$i$, equal~$x^\star$ and~$\hat{x}$ if we choose $i = i(x^\star, \hat{x})$. For other choices of~$i$, this is, in general, not the case.

Observe that all possible realizations of~$w_2$ from the line $L \DEF \SET{ w_2 + x \cdot a_i \WHERE x \in \RR }$ are mapped to the same vector $\tilde{w}_2(w_2, i)$. Consequently, if~$w_1$ is fixed and if we only consider realizations of~$\mu$ for which $w_2 \in L$, then vertex~$\tilde{p}^\star$ and, hence, vertex~$y^\star$ from Definition~\ref{definition:event E} are already determined. However, since~$w_2$ is not completely specified, we have some randomness left for event $\E_{i, t, \e}$ to occur. This allows us to bound the probability of event $\E_{i, t, \e}$ from above (see proof of Lemma~\ref{lemma:probability bound}). The next lemma shows why this probability matters.

\begin{lemma}
\label{lemma:event covering}
For reals $t \geq 0$ and $\e > 0$ let $\A_{t, \e}$ denote the event that the path $R = R_{w_1, w_2}$ has a slope in $(t, t+\e]$. Then, $\A_{t, \e} \subseteq \bigcup_{i=1}^m \E_{i, t, \e}$.
\end{lemma}

\begin{proof}
Assume that event $\A_{t, \e}$ occurs. Let~$\hat{p}$ be the right-most vertex of~$R$ whose slope exceeds~$t$, let~$p^\star$ be the right neighbor of~$\hat{p}$, and let~$\hat{x}$ and~$x^\star$ be the neighboring vertices of~$P$ for which $\pi(\hat{x}) = \hat{p}$ and $\pi(x^\star) = p^\star$, where $\pi = \pi_{w_1, w_2}$. Moreover, let $i = i(x^\star, \hat{x})$ be the index for which $a_i^\T x^\star = b_i$ but $a_i^\T \hat{x} < b_i$. We show that event~$\E_{i, t, \e}$ occurs.

Consider the left-most vertex~$\tilde{p}^\star$ of $\tilde{R} = R_{w_1, \tilde{w}_2(w_2, i)}$ whose slope does not exceed~$t$ and let~$y^\star$ be the corresponding vertex of~$P$. In accordance with Corollary~\ref{corollary:reconstruct} we obtain $y^\star = x^\star$. Hence, $a_i^\T y^\star = b_i$, i.e., the first condition of event~$\E_{i, t, \e}$ holds. Now let~$\hat{y}$ be the unique neighbor~$y$ of~$y^\star$ for which $a_i^\T y < b_i$. Since $y^\star = x^\star$, we obtain $\hat{y} = \hat{x}$. Consequently,
\[
  \frac{w_2^\T (\hat{y} - y^\star)}{w_1^\T (\hat{y} - y^\star)}
  = \frac{w_2^\T (\hat{x} - x^\star)}{w_1^\T (\hat{x} - x^\star)}
  \in (t, t+\e] \COMMA
\]
since this is the smallest slope of~$R$ that exceeds~$t$ and since there is a slope in $(t, t+\e]$ by assumption. Hence, event $\E_{i, t, \e}$ occurs since the second condition for event $\E_{i, t, \e}$ to happen holds as well.
\end{proof}

With Lemma~\ref{lemma:event covering} we can now bound the probability of event $\A_{t, \e}$.

\begin{lemma}
\label{lemma:probability bound}
For reals $t \geq 0$ and $\e > 0$ the probability of event $\A_{t, \e}$ is bounded by
$
  \Pr{\A_{t, \e}} \leq \frac{4mn^2\e}{\delta^2}
$.
\end{lemma}

\begin{proof}
Due to Lemma~\ref{lemma:event covering} it suffices to show that
$
  \Pr{\E_{i, t, \e}}
  \leq \frac{1}{m} \cdot \frac{4mn^2\e}{\delta^2}
  = \frac{4n^2\e}{\delta^2}
$
for any index $i \in [m]$.

We apply the principle of deferred decisions and assume that vector $\lambda \in (0, 1]^n$ is not random anymore, but arbitrarily fixed. Thus, vector~$w_1$ is already fixed. Now we extend the normalized vector~$a_i$ to an orthonormal basis $\SET{ q_1, \ldots, q_{n-1}, a_i }$ of~$\RR^n$ and consider the random vector $(Y_1, \ldots, Y_{n-1}, Z)^\T = Q^\T w_2$ given by the matrix vector product of the transpose of the orthogonal matrix $Q = [q_1, \ldots, q_{n-1}, a_i]$ and the vector $w_2 = [v_1, \ldots, v_n] \cdot \mu$. For fixed values $y_1, \ldots, y_{n-1}$ let us consider all realizations of~$\mu$ such that $(Y_1, \ldots, Y_{n-1}) = (y_1, \ldots, y_{n-1})$. Then,~$w_2$ is fixed up to the ray
\vspace{-0.5em}\[
  w_2(Z)
  = Q \cdot (y_1, \ldots, y_{n-1}, Z)^\T
  = \sum_{j=1}^{n-1} y_j \cdot q_j + Z \cdot a_i
  = w + Z \cdot a_i \vspace{-0.5em}
\]
for $w = \sum_{j=1}^{n-1} y_j \cdot q_j$. All realizations of $w_2(Z)$ that are under consideration are mapped to the same value~$\tilde{w}_2$ by the function $w_2 \mapsto \tilde{w}_2(w_2, i)$, i.e., $\tilde{w}_2(w_2(Z), i) = \tilde{w}_2$ for any possible realization of~$Z$. In other words, if $w_2 = w_2(Z)$ is specified up to this ray, then the path $R_{w_1, \tilde{w}_2(w_2, i)}$ and, hence, the vectors~$y^\star$ and~$\hat{y}$ used for the definition of event $\E_{i, t, \e}$, are already determined.

Let us only consider the case that the first condition of event~$\E_{i, t, \e}$ is fulfilled. Otherwise, event~$\E_{i, t, \e}$ cannot occur. Thus, event $\E_{i, t, \e}$ occurs iff
\[
  (t, t+\e]
  \ni \frac{w_2^\T \cdot (\hat{y} - y^\star)}{w_1^\T \cdot (\hat{y} - y^\star)}
  = \underbrace{\frac{w^\T \cdot (\hat{y} - y^\star)}{w_1^\T \cdot (\hat{y} - y^\star)}}_{\FED \alpha} + Z \cdot \underbrace{\frac{a_i^\T \cdot (\hat{y} - y^\star)}{w_1^\T \cdot (\hat{y} - y^\star)}}_{\FED \beta} \DOT
\]
The next step in this proof will be to show that the inequality $|\beta| \geq \frac{\delta}{n}$ is necessary for event~$\E_{i, t, \e}$ to happen. For the sake of simplicity let us assume that $\|\hat{y} - y^\star\| = 1$ since~$\beta$ is invariant under scaling. If event~$\E_{i, t, \e}$ occurs, then $a_i^\T y^\star = b_i$, $\hat{y}$ is a neighbor of~$y^\star$, and $a_i^\T \hat{y} \neq b_i$. That is, by Lemma~\ref{lemma:delta properties}, Claim~\ref{delta properties:neighboring vertices} we obtain $|a_i^\T \cdot (\hat{y} - y^\star)| \geq \delta \cdot \|\hat{y} - y^\star\| = \delta$ and, hence,
\[
  |\beta|
  = \left| \frac{a_i^\T \cdot (\hat{y} - y^\star)}{w_1^\T \cdot (\hat{y} - y^\star)} \right|
  \geq \frac{\delta}{|w_1^\T \cdot (\hat{y} - y^\star)|}
  \geq \frac{\delta}{\|w_1\| \cdot \|\hat{y}-y^\star)\|}
  \geq \frac{\delta}{n \cdot 1} \DOT
\]
Summarizing the previous observations we can state that if event~$\E_{i, t, \e}$ occurs, then $|\beta| \geq \frac{\delta}{n}$ and $\alpha + Z \cdot \beta \in (t, t+\e] \subseteq [t-\e, t+\e]$. Hence,
\[
  Z
  \in \left[ \frac{t-\alpha}{\beta} - \frac{\e}{|\beta|}, \frac{t-\alpha}{\beta} + \frac{\e}{|\beta|} \right]
  \subseteq \left[ \frac{t-\alpha}{\beta} - \frac{\e}{\frac{\delta}{n}}, \frac{t-\alpha}{\beta} + \frac{\e}{\frac{\delta}{n}} \right]
  \FED I(y_1, \ldots, y_{n-1}) \DOT
\]
Let~$\B_{i, t, \e}$ denote the event that~$Z$ falls into the interval $I(Y_1, \ldots, Y_{n-1})$ of length $\frac{2n \e}{\delta}$. We showed that $\E_{i, t, \e} \subseteq \B_{i, t, \e}$. Consequently,
\[
  \Pr{\E_{i, t, \e}}
  \leq \Pr{\B_{i, t, \e}}
  \leq \frac{2n \cdot \frac{2n \e}{\delta}}{\delta(Q^\T v_1, \ldots, Q^\T v_n)}
  \leq \frac{4n^2\e}{\delta^2} \COMMA
\]
where the second inequality is due to first claim of Theorem~\ref{theorem.Prob:enough randomness}: By definition, we have
\[
  (Y_1, \ldots, Y_{n-1}, Z)^\T
  = Q^\T w_2
  = Q^\T \cdot [v_1, \ldots, v_n] \cdot \mu
  = [Q^\T v_1, \ldots, Q^\T v_n] \cdot \mu \DOT
\]
The third inequality stems from the fact that
$
  \delta(Q^\T v_1, \ldots, Q^\T v_n)
  = \delta(v_1, \ldots, v_n)
  \geq \delta
$,
where the equality is due to the orthogonality of~$Q$ (Claim~\ref{delta properties:orthogonal matrix} of Lemma~\ref{lemma:delta properties}).
\end{proof}

\begin{lemma}
\label{lemma:expectation bound}
Let~$Y$ be the number of slopes of $R = R_{w_1,w_2}$ that lie in the interval $(0, 1]$. Then, $\Ex{Y} \leq \frac{4mn^2}{\delta^2}$.
\end{lemma}

\begin{proof}
For a real $\e > 0$ let~$\F_\e$ denote the event from Definition~\ref{definition:failure event}. Recall that all slopes of~$R$ differ by more than~$\e$ if~$\F_\e$ does not occur. Let~$Z_{t,\e}$ be the random variable that indicates whether~$R$ has a slope in the interval $(t, t+\e]$ or not, i.e., $Z_{t,\e} = 1$ if there is such a slope and $Z_{t,\e} = 0$ otherwise. Then, for any integer $k \geq 1$
\[
  Y \leq \begin{cases}
    \sum_{i=0}^{k-1} Z_{\frac{i}{k}, \frac{1}{k}} & \text{if $\F_{\frac{1}{k}}$ does not occur} \COMMA \cr
    m^n & \text{otherwise} \DOT
  \end{cases}
\]
This is true since $\binom{m}{n-1} \leq m^n$ is a worst-case bound on the number of edges of~$P$ and, hence, of the number of slopes of~$R$. Consequently,
\begin{align*}
  \Ex{Y}
  &\leq \sum_{i=0}^{k-1} \Ex{Z_{\frac{i}{k}, \frac{1}{k}}} + \Pr{\F_{\frac{1}{k}}} \cdot m^n
  = \sum_{i=0}^{k-1} \Pr{\A_{\frac{i}{k}, \frac{1}{k}}} + \Pr{\F_{\frac{1}{k}}} \cdot m^n \cr
  &\leq \sum_{i=0}^{k-1} \frac{4mn^2 \cdot \frac{1}{k}}{\delta^2} + \Pr{\F_{\frac{1}{k}}} \cdot m^n
  = \frac{4mn^2}{\delta^2} + \Pr{\F_{\frac{1}{k}}} \cdot m^n \COMMA
\end{align*}
where the second inequality stems from Lemma~\ref{lemma:probability bound}. The claim follows since the bound on $\Ex{Y}$ holds for any integer $k \geq 1$ and since $\Pr{\F_\e} \to 0$ for $\e \to 0$ in accordance with Lemma~\ref{lemma:failure probability II}.
\end{proof}

\begin{proof}[Proof of Theorem~\ref{maintheorem}]
Lemma~\ref{lemma:expectation bound} bounds only the expected number of edges on the path~$R$ that
have a slope in the interval~$(0,1]$. However, the lemma can also be used to bound the expected number
of edges whose slope is larger than~$1$. For this, one only needs to exchange the order
of the objective functions~$w_1^\T x$ and~$w_2^\T x$ in the projection~$\pi$. Then any edge with a
slope of~$s>0$ becomes an edge with slope~$\frac{1}{s}$. Due to the symmetry in the choice of~$w_1$ and~$w_2$,
Lemma~\ref{lemma:expectation bound} can also be applied to bound the expected
number of edges whose slope lies in~$(0,1]$ for this modified projection, which are exactly the
edges whose original slope lies in~$[1,\infty)$. 

Formally we can argue as follows.
Consider the vertices $x'_1 = x_2$ and $x'_2 = x_1$, the directions $w'_1 = -w_2$ and $w'_2 = -w_1$, and the projection $\pi' = \pi_{w'_1,w'_2}$, yielding a path~$R'$ from $\pi'(x'_1)$ to $\pi'(x'_2)$. Let~$X$ be the number of slopes of~$R$ and let~$Y$ and~$Y'$ be the number of slopes of~$R$ and of~$R'$, respectively, that lie in the interval $(0, 1]$. The paths~$R$ and~$R'$ are identical except for the linear transformation
$
  \begin{bmatrix} 
    x \cr
    y
  \end{bmatrix}
  \mapsto
  \begin{bmatrix}
    0 & -1 \cr
    -1 & 0
  \end{bmatrix}
  \cdot
  \begin{bmatrix}
    x \cr
    y
  \end{bmatrix}
$.
Consequently, $s$ is a slope of~$R$ if and only if $\frac{1}{s}$ is a slope of~$R'$ and, hence, $X \leq Y + Y'$. One might expect equality here but in the unlikely case that~$R$ contains an edge with slope equal to~$1$ we have $X = Y + Y' - 1$. The expectation of~$Y$ is given by Lemma~\ref{lemma:expectation bound}. Since this result holds for any two vertices~$x_1$ and~$x_2$ it also holds for~$x'_1$ and~$x'_2$. Note, that~$w'_1$ and~$w'_2$ have exactly the same distribution as the directions the shadow vertex algorithm computes for~$x'_1$ and~$x'_2$. Therefore, Lemma~\ref{lemma:expectation bound} can also be applied to bound~$\Ex{Y'}$ and we obtain $\Ex{X} \leq \Ex{Y} + \Ex{Y'} = \frac{8mn^2}{\delta^2}$.
\end{proof}

The proof of Corollary~\ref{maincorollaryI} follows immediately from Theorem~\ref{maintheorem} and Claim~\ref{delta properties:comparison} of Lemma~\ref{lemma:delta properties}.

\section{Some Probability Theory}
\label{sec:protheory}

\newcommand{\Src}{\cite{BrunschR12}\xspace}
\newcommand{\ClaimI}{Theorem~35\xspace}
\newcommand{\ClaimII}{Lemma~36\xspace}
\newcommand{\ClaimIII}{Corollary~37\xspace}

The following theorem is a variant of \ClaimI from~\Src. The two differences are as follows: In~\Src arbitrary densities are considered. We only consider uniform distributions. On the other hand, instead of considering matrices with entries from $\SET{ -1, 0, 1 }$ we consider real-valued square matrices. This is why the results from~\Src cannot be applied directly.

\begin{theorem}
\label{theorem.Prob:enough randomness}
Let~$X_1, \ldots, X_n$ be independent random variables uniformly distributed on $(0, 1]$, let $A = [a_1, \ldots, a_n] \in \RR^{n \times n}$ be an invertible matrix, let $(Y_1, \ldots, Y_{n-1}, Z)^\T = A \cdot (X_1, \ldots, X_n)^\T$ be the linear combinations of~$X_1, \ldots, X_n$ given by~$A$, and let $I \colon \RR^{n-1} \to \SET{ [x, x+\e] \WHERE x \in \RR }$ be a function mapping a tuple $(y_1, \ldots, y_{n-1})$ to an interval $I(y_1, \ldots, y_{n-1})$ of length~$\e$. Then the probability that~$Z$ lies in the interval $I(Y_1, \ldots, Y_{n-1})$ can be bounded by
\[
  \Pr{Z \in I(Y_1, \ldots, Y_{n-1})} \leq \frac{2n\e}{\delta(a_1, \ldots, a_n) \cdot \min_{k \in [n]} \|a_k\|} \DOT
\]
\end{theorem}

\begin{proof}
Let us consider the proof of \ClaimI of~\Src for $m = n$ and $k = 1$. We obtain
\[
  \Pr{Z \in I(Y_1, \ldots, Y_{n-1})} \leq \e \cdot |\det(A^{-1})| \cdot \int_{y \in \RR^{n-1}} \max_{z \in \RR} f_X(A^{-1} \cdot (y, z)^\T) \d y \COMMA
\]
where~$f_X$ denotes the common density of the variables $X_1, \ldots, X_n$. In our case, $f_X$ is~$1$ on $(0, 1]^n$ and~$0$ otherwise. Note, that in the proof of \ClaimI matrix~$A$ was an integer matrix and so $|\det(A^{-1})| \leq 1$. In this proof considering this factor is crucial.

It remains to bound $\int_{y \in \RR^{n-1}} \max_{z \in \RR} f_X(A^{-1} \cdot (y, z)^\T) \d y$. For this we only have to consider the proof of \ClaimII of~\Src since all densities are rectangular functions. Here, we have $\chi = 1$ and $\ell_i = 1$ and $\phi_i = 1$ for any $i \in [n]$. The only point where the structure of matrix~$A$ is exploited is where $|\det(PAT)|$ for $P = [\ID{n-1}, \ZERO{n-1}{1}]$ and $T = [e_1, \ldots, e_{i-1}, e_{i+1}, \ldots, e_n]$ for an arbitrary index $i \in [n]$ is bounded. Since $PAT = A_{n, i}$, we obtain
\begin{align*}
  \int_{y \in \RR^{n-1}} \max_{z \in \RR} f_X(A^{-1} \cdot (y, z)^\T) \d y
  &\leq \chi \cdot \sum_{i \in [n]} \sum_{j=0}^1 |\det(A_{n, i})| \cdot \prod_{i' \neq i} \ell_{i'} \cr
  &= 2 \cdot \sum_{i \in [n]} |\det(A_{n, i})| \DOT
\end{align*}
Summarizing both bounds, we obtain
\[
  \Pr{Z \in I(Y_1, \ldots, Y_{n-1})}
  \leq 2\e \cdot |\det(A^{-1})| \cdot \sum_{i \in [n]} |\det(A_{n, i})|
  = 2\e \cdot \sum_{i \in [n]} \frac{|\det(A_{n, i})|}{|\det(A)|} \DOT
\]
We are now going to bound the fraction $\frac{|\det(A_{n, i})|}{|\det(A)|}$. To do this, consider the equation $Ax = e_n$. We obtain
\[
  |x_i|
  = \frac{|\det([a_1, \ldots, a_{i-1}, e_n, a_{i+1}, \ldots, a_n])|}{|\det(A)|}
  = \frac{|\det(A_{n,i})|}{|\det(A)|} \COMMA
\]
where the first equality is due to Cramer's rule and the second equality is due to Laplace's formula. Hence,
\[
  \Pr{Z \in I(Y_1, \ldots, Y_{n-1})}
  \leq 2\e \cdot \sum_{i \in [n]} \frac{|\det(A_{n, i})|}{|\det(A)|}
  = 2\e \cdot \|x\|_1
  \leq 2\sqrt{n}\e \cdot \|x\|_2
\]
Now consider the equation $\hat{A} \hat{x} = e_n$ for $\hat{A} = [\N(a_1), \ldots, \N(a_n)]$. Vector $\hat{x} = \hat{A}^{-1} e_n$ is the $n^\text{th}$ column of the matrix $\hat{A}^{-1}$. Thus, we obtain
\[
  \|\hat{x}\|
  \leq \max_{\substack{\text{$r$ column}\\\text{of $\hat{A}^{-1}$}}} \|r\|
  \leq \frac{\sqrt{n}}{\delta(a_1, \ldots, a_n)} \COMMA
\]
where second inequality is due to Claim~\ref{delta properties:inverse} of Lemma~\ref{lemma:delta properties}. Due to $A = \hat{A} \cdot \diag(\|a_1\|, \ldots, \|a_n\|)$, we have
\[
  x
  = A^{-1} e_n
  = \diag \left( \frac{1}{\|a_1\|}, \ldots, \frac{1}{\|a_n\|} \right) \cdot \hat{A}^{-1} e_n
  = \diag \left( \frac{1}{\|a_1\|}, \ldots, \frac{1}{\|a_n\|} \right) \cdot \hat{x} \DOT
\]
Consequently, $\|x\| \leq \frac{\|\hat{x}\|}{\min_{k \in [n]} \|a_k\|}$ and, thus,
\begin{align*}
  \Pr{Z \in I(Y_1, \ldots, Y_{n-1})}
  &\leq 2\sqrt{n}\e \cdot \frac{\|\hat{x}\|}{\min_{k \in [n]} \|a_k\|} \cr
  &\leq \frac{2n\e}{\delta(a_1, \ldots, a_n) \cdot \min_{k \in [n]} \|a_k\|} \DOT \qedhere
\end{align*}
\end{proof}

%\bibliographystyle{plain}
%\bibliography{literature}

\end{document}